\documentclass[11pt,onecolumn,a4paper]{article}

\usepackage{amsmath}
\usepackage{amsfonts}
\usepackage{amssymb}
\usepackage{amsthm}
\usepackage{natbib}
\usepackage{graphicx,color}
\usepackage{placeins}
\usepackage{url}
\usepackage{latexsym,enumerate}
\usepackage{booktabs}
\usepackage{framed}
\usepackage{subfigure}
\usepackage{multirow} 
\usepackage{color}
\usepackage[table]{xcolor}
\usepackage{tikz}				
\usetikzlibrary{shapes,snakes}
\usetikzlibrary{shapes.multipart}
\usetikzlibrary{fit}	
\usepackage{colortbl}
\usepackage{todonotes}

\definecolor{gray1}{rgb}{0.8,0.8,0.8}
\definecolor{gray2}{rgb}{0.95,0.95,0.95}
\definecolor{light-gray}{gray}{0.95}
\definecolor{light-red}{rgb}{0.96, 0.76, 0.76}
\definecolor{light-blue}{rgb}{0.63, 0.79, 0.95}
\newcommand{\overbar}[1]{\mkern 1.5mu\overline{\mkern-1.5mu#1\mkern-1.5mu}\mkern 1.5mu}

\newtheorem{Th}{Theorem}[section]

\newtheorem{Rm}[Th]{Remark}
\newtheorem{Algo}[Th]{Algorithm}
\newtheorem{Test}[Th]{Test}

\newcommand{\RE}{\,{\rm Re}}

\newcommand{\CC}{{\mathbb C}}
\newcommand{\NN}{{\mathbb N}}

\newcommand{\Prb}{{\mathbb P}}

\newcommand{\RR}{{\mathbb R}}
\newcommand{\SSS}{{\mathbb S}}
\newcommand{\cN}{{\mathcal{N}}}
\DeclareMathOperator*{\argmin}{argmin}

\DeclareMathOperator*{\Arg}{Arg}
\newcommand{\iid}{\operatorname{\stackrel{i.i.d.}{\sim}}}

\newcommand\eqnref[1]{Equation~(\ref{#1})}
\newcommand\citeegp[1]{(e.g. \cite{#1})}

\begin{document}

\title{Detecting Anisotropy in Fingerprint Growth}

\author{Karla Markert\footnote{Georg-August-Universit\"at G\"ottingen, Institut f\"ur Mathematische Stochastik}, Karolin Krehl$^*$, Carsten Gottschlich$^*$, \\and Stephan Huckemann$^*$\footnote{huckeman@math.uni-goettingen.de}}

\date{}
\maketitle

\begin{abstract}
    From infancy to adulthood, human growth is anisotropic, much more along the proximal-distal axis (height)  than along the medial-lateral axis  (width), particularly at extremities. 
    Detecting and modeling the rate of anisotropy in fingerprint growth, and possibly other growth patterns as well, facilitates the use of children's fingerprints for long-term biometric identification. Using standard fingerprint scanners, anisotropic growth is highly overshadowed by the varying distortions created by each imprint, and it seems that this difficulty has  hampered to date the development of suitable methods, detecting anisotropy, let alone, designing models. 

    We provide a tool chain to statistically detect, with a given confidence, anisotropic growth in fingerprints and its preferred axis,
    where we only require a standard fingerprint scanner and a minutiae matcher.
    We build on a perturbation model, a new Procrustes-type algorithm, use and develop several parametric and non-parametric tests for different hypotheses, in particular for neighborhood hypotheses to detect the axis of anisotropy, where the latter tests are tunable to measurement accuracy. 
    Taking into account realistic distortions caused by pressing fingers on scanners, our simulations based on real data  indicate that, for example, already in rather small samples 
    (56 matches)
    we can  significantly detect 
    proximal-distal growth if it exceeds medial-lateral growth by only around $5\,\%$. 
    
    Our method is well applicable to future datasets of children fingerprint time series and we provide an implementation of our algorithms and tests with matched minutiae pattern data.
\end{abstract}

\section*{Keywords} Circular statistics; fingerprint minutiae; von Mises distribution; 
extrinsic mean; Procrustes analysis; neighborhood hypothesis, consistency bias

\section{Introduction}

When identification of humans by identity documents is not reliably possible, it is common practice to fall back on identification by biometric traits, as for instance in the monumental Aadhaar program initiated by the Unique Identification  
Authority of India\footnote{https://uidai.gov.in/} (UIDAI) in 2009. Among these, automated fingerprint identification systems (AFIS) have proven to be highly successful, because fingerprints are easily accessible and AFISs only requires minimal infrastructure \citeegp{Sonla2007using}. 

Indeed, human recognition by fingerprints has evolved into a mature technology, for an overview see \cite{HandbookFingerprintRecognition2009}. Identifying suspects in forensic investigations by fingermarks left at crime scenes has been the main field of application of fingerprint recognition over the past century. An important governmental application of fingerprint recognition is border control. Nowadays also many commercial applications rely on fingerprint recognition: Hundreds of millions of people use fingerprints to unlock their smartphone or tablet PC, and increasingly, fingerprints are also used for authorizing financial transactions in mobile payment applications. 
Despite all the progress in the development of algorithms and hardware, a number of challenges remain to be addressed and solved. This includes the development of methods for processing low-quality and very low-quality images \citeegp{Schumacher2013}), especially latent fingerprints \citeegp{SankaranVatsaSingh2014}, and for tasks like image segmentation \citeegp{ThaiHuckemannGottschlich2016} or image enhancement \citeegp{Gottschlich2012,GottschlichSchoenlieb2012}. Further topics which require more research are fingerprint liveness detection \citeegp{Gottschlich2016} as a countermeasure to presentation attacks with spoof fingerprints, and fingerprint recognition for toddlers and children \citeegp{Schumacher2013}.

While most AFISs are designed for adult fingerprints, specific challenges arise when identifying smaller children, even newborns, in vaccination programs, say. Such challenges (e.g. much smaller prints, usually of much lower quality with much higher ridge frequencies) have been recently addressed and considerable progress has been made, e.g. \cite{KotzerkeArakalaDavisHoradamMcVernon2014, jain2015biometrics,Jain2017fingerprint}. These challenges have also led to the development of authentication methods for newborns, say, based on derived features such as key points from pore patterns \citeegp{LemesSecundoBellonSilva2014} or of entire palm- or footprints \citeegp{LemesBellonSilvaJain2011,jia2012newborn}).

For adults, fingerprints can also be used for very long-term identification, because, as \cite{Galton1892} observed, fingerprints remain largely unchanged throughout adulthood, unless seriously damaged. Obviously this is wrong for fingerprints of children, because between infancy and adulthood, body-size grows by a considerable factor, where, in particular for the limbs, distal-proximal (length) growth exceeds by far the lateral-medial (width) growth. While there is consensus that the topological structure of fingerprint ridge lines, in particular singular points and minutiae, is fixated long before birth  around the 12th gestational week \citeegp{KuckenNewell2007}, 
it is unclear to date, how growth, either over the entire growing period or parts thereof, distorts this structure. In fact, as fingerprints are on an extremal locus on an extremity, it is natural to expect a considerable rate of anisotropic growth. As a first study to this end,  \cite{HotzGottschlichLorenzBernhardtHantschelMunk2011} derived a framework for adolescent fingerprints, modeling isotropic growth correlated with body size, thereby greatly improving forensic matching rates for the German Federal Criminal Bureau (BKA), cf. \cite{GottschlichHotzLorenzBernhardtHantschelMunk2011}. While it is desirable to draw conclusions regarding anisotropic growth ``in
order to give a clear message to developers of fingerprint recognition systems'', as demanded by the \cite{JRC2013}, this  study concluded that possible anisotropic effects, at least for adolescents, seem to be overshadowed by the variability of local distortions caused by physically placing  fingers onto a scanner and by bad quality of imprints leading to minutiae mismatches. Figure \ref{fig:distortion} shows such typical distortions including those due to imprecise minutiae recognition, even on imprints of fairly good quality.

	\begin{figure}[h!]
	\centering
	\includegraphics[width=0.27\textwidth, clip=true, trim= 0cm 0cm 0cm 0cm]{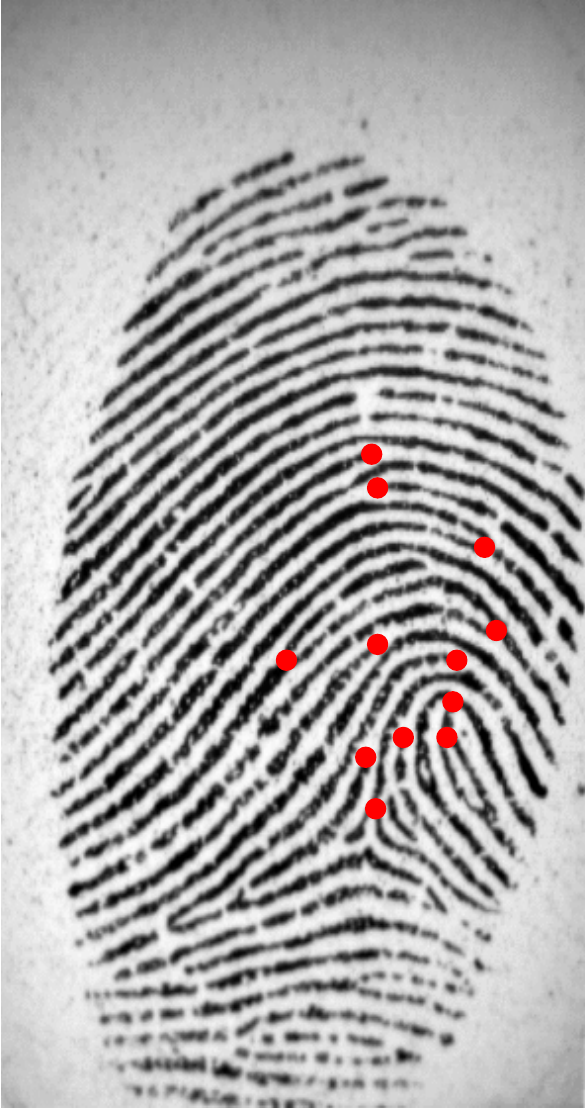} 
	\includegraphics[width=0.36\textwidth]{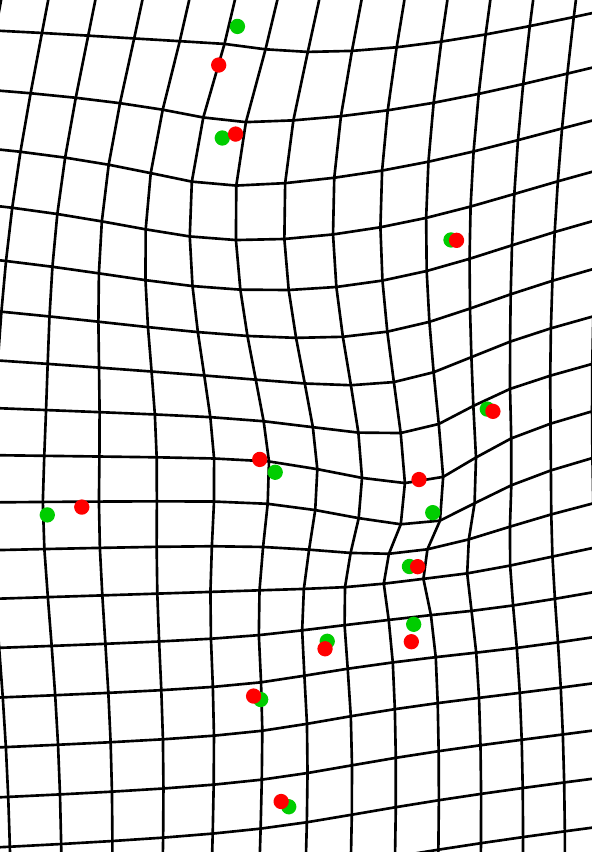}
	\includegraphics[width=0.27\textwidth, clip=true, trim= 0cm 0cm 0cm 0cm]{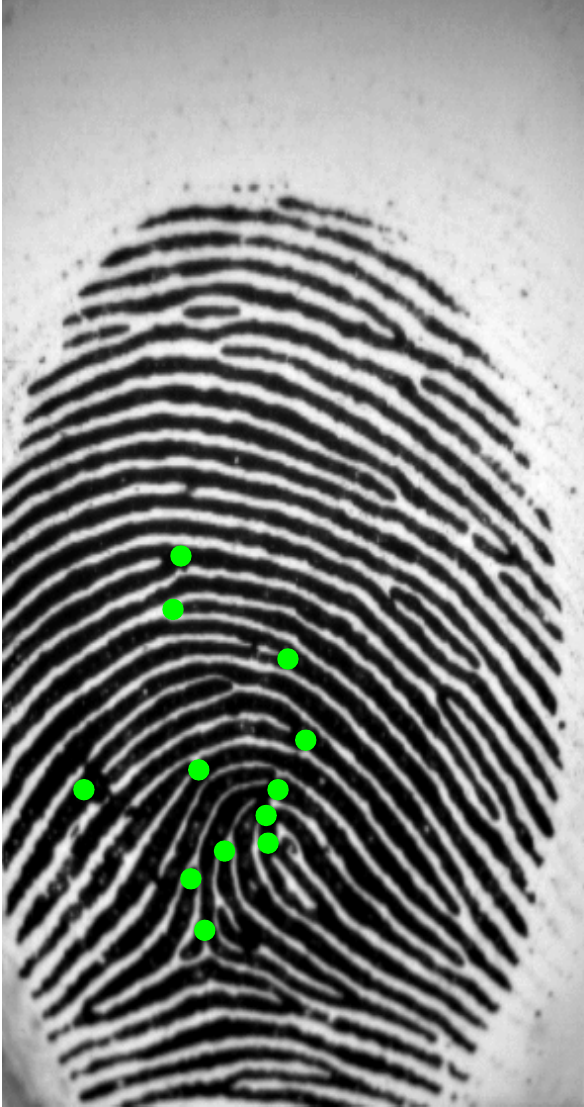} 
	\caption{\it Imprint 2 (left) and imprint 6 (right) of the second finger from FVC2002 DB2 with marked minutiae and closeup of overlaid minutiae (center) along with relative distortion (which may also be caused by imprecise minutiae recognition, e.g. the top minutia and the third minutia from the right) estimated by a thin-plate spline model using \cite{Dshapes}.\label{fig:distortion}}
	\end{figure}

As a first step towards the design of general fingerprint growth models, in this paper, we 
develop a series of statistical tests for anisotropic growth of fingerprint minutiae patterns, cf. Table \ref{table:TestOverview}. As shown in Figure \ref{PictureTest} at the core of our work is a Procrustes-type \citeegp{DM98} algorithm, to estimate growth and nuisance values for the minutiae point patterns. Then, various statistical tests are introduced to analyze the distributions of the parameters governing anisotropic growth. Our first three tests detect anisotropy and we discuss the power of the first two. More precisely, by simulation we provide for an estimate of minimal anisotropy that can be detected, in the presence of natural distortions due to finger placement on scanners. For the number of $56$ matches considered, it turns out that already $2.6\,\%$ (Test \ref{test:Rayleigh}) or $6.1\,\%$ (Test \ref{testfortau}), respectively, of anisotropy can be statistically detected. If anisotropy has been detected by these first tests, our next tests can be applied to detect the direction of growth, and we give a parametric and a nonparametric version. For each we show that they (asymptotically) keep the level and asymptotically their false rejection rate goes to zero while their power tends to one. We validate these tests by using real distortions and simulated growth and illustrate that we can detect the direction of growth within the accuracy of aligning fingerprints. 
All of the simulations are based on hand marked data from the FVC 2002 DB2 \citep{FVC2002}.

More precisely, for the 8 imprints of fingers 2 -- 9, we have have used VeriFinger\footnote{www.neurotechnology.com/} to extract and match minutiae, we have manually corrected obvious mismatches and we have stored minutiae loci in the data files of the supplemental package. For every one of the seven imprints, 2 -- 8, we have stored a copy of the first imprint featuring corresponding minutiae.

\begin{center}
	\begin{figure}[h!]

	\footnotesize{
	\centering
		\begin{tikzpicture}[every text node part/.style={align=center}]
		\node (Real) [ellipse, draw, fill=light-gray, scale=1.2] {Unknown true values \\ ~\\\fbox{
		\begin{tabular}{r|l}
		$ \beta$ & fingerprint rotation angle\\
		$ \lambda$ & rate of isotropic growth\\
		$ \gamma $ & angle of anisotropic growth\\
		$ \tau$ & rate of anisotropic growth
		\end{tabular}
		}
		}; 
		\node (Calc) [below = 4 cm] [ellipse, draw, fill=light-gray, scale=1.2] {Estimated values \\ $\hat{\beta},  \hat{\lambda}, \hat{\gamma}, \hat{\tau}$};
		\node (Test) [below = 7cm] [ellipse, draw, fill=light-gray, scale=1.2] {Distributions of \\ $\hat{\gamma}$ and $\hat{\tau}$};
		\node[draw,inner sep=2mm,label=right:{\normalsize Algorithm 2.1},fit=(Real) (Calc), , fill=light-red, fill opacity=0.15] {};
		\node[draw,inner sep=2mm,label=right:{\normalsize Tests 3.1 -- 3.5},fit=(Calc) (Test), scale=1.1, fill=light-blue, fill opacity=0.1] {};
		\draw[->, thick] (Real) to (Calc);
		\draw[->, thick] (Calc) to (Test);
		\end{tikzpicture}}
		\caption{\it Nuisance parameter $\beta$ and values $\lambda,\gamma$ and $\tau$ governing growth of fingerprint minutiae point patterns are estimated by Algorithm \ref{mainalgo}. The five tests listed in Table \ref{table:TestOverview} are based on the distributions of suitable estimates.}
		\label{PictureTest}
	\end{figure}
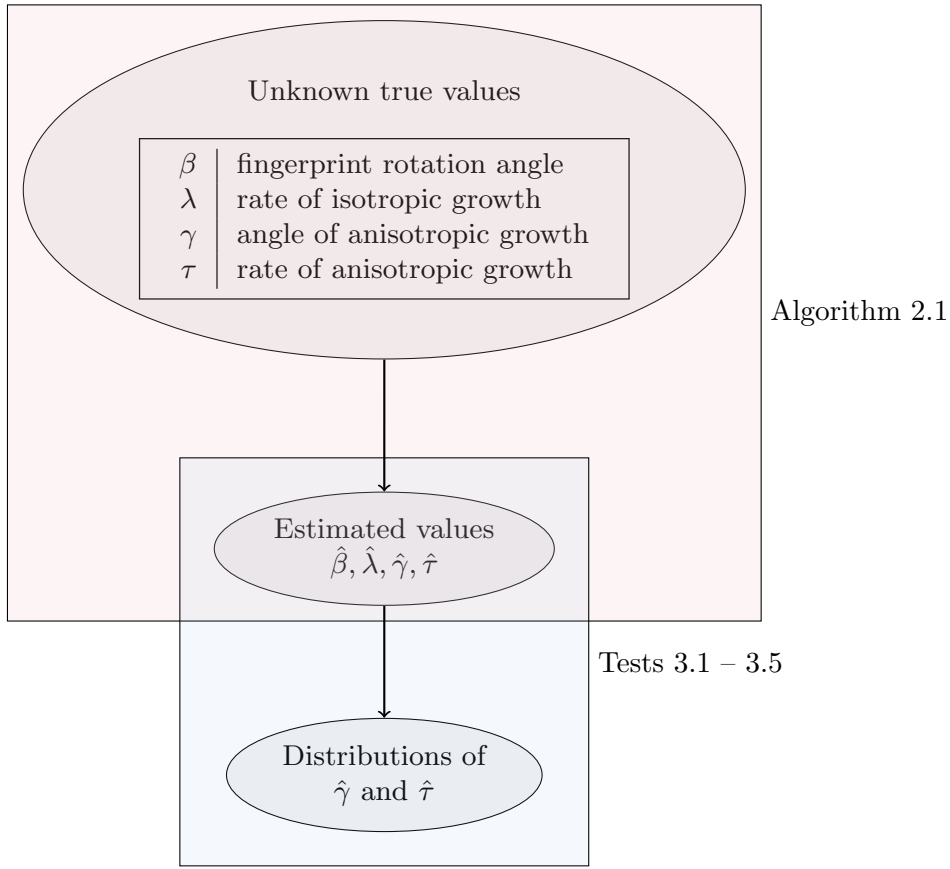

\end{center}

      \begin{table}[ht!]
      	\begin{center}
      		\begin{tabular}{|l|l|l|}
      			\hline \hline 
      			{Test\cellcolor[gray]{0.8}} & {Detail \cellcolor[gray]{0.8}} & {$H_0$\cellcolor[gray]{0.8}} \\
      			\hline \hline
      			\rowcolor{light-gray} Test \ref{test:Rayleigh} 
      			& parametric test for $\hat{\gamma}$ & \\
      			\rowcolor{light-gray} Test \ref{testfortau} 
      			& simulation test for $\hat{\tau}$ & \\
      			\rowcolor{light-gray} Test \ref{test:jointly} 
      			& simulation test jointly for $\hat{\gamma}$ and $\hat{\tau}$ & \multirow{-3}{*}{growth is isotropic} \\
      			\rowcolor{white} Test \ref{vM-anisotropy:test} 
      			& parametric test for distal growth & \\
      			\rowcolor{white} Test \ref{nonparametric-anisotropy:test} 
      			& nonparametric test for distal growth & \multirow{-2}{*}{growth is not distal}\\
      			\hline
      		\end{tabular}
      		\caption{\it Tests for the distributions of $\hat{\gamma}$ and $\hat{\tau}$ and null hypotheses $H_0$ to be rejected.}
      	\label{table:TestOverview}
      \end{center}
 	 \end{table} 

As our growth model can be viewed as a perturbation (or generative) model, it is not clear that estimation via a Procrustes-type algorithm will recover the modeled ``true'' anisotropic growth, at least asymptotically. In fact, it is well known from shape analysis, that under general error models, Procrustes means are usually \emph{inconsistent} with centers of perturbation models, cf. \cite{Lele93,Le98,KM97,H_Procrustes_10,DevilliersAllassonniereTrouvePenec2017}. Indeed, our simulations based on real distortions in Section \ref{scn:validation-simulation-schemes} indicate toward minor inconsistency, namely that true anisotropic growth is slightly overestimated and this effect wanes with increased anisotropy. Quantifying inconsistency in detail, as recently done by \cite{MiolaneHolmesPennec2017}, say, is beyond the scope of this paper.

Our paper is organized as follows. In Section \ref{scn:algo} we propose our algorithm to estimate anisotropic growth. Section \ref{scn:tests} then introduces our validation and simulation schemes, five tests, and we validate and simulate sensitivity of four tests. Finally, in Section \ref{scn:fidelity} we discuss the consistency of our algorithm for the growth model and, along with it, also simulate sensitivity for scenarios of varying unknown growth. We conclude with an outlook to future applications. 

We provide an R-package\footnote{\url{
http://www.stochastik.math.uni-goettingen.de/AnisotropicGrowth_0.1.3.tar.gz}} implementing our algorithms and tests, including the manually marked and matched minutiae pattern data. Code producing all computations, simulations and graphics presented can be found on the man pages of corresponding routines.

\section{Modeling and Estimating Anisotropy}\label{scn:algo}

\subsection{A Model for 
Uni-Directional Anisotropic Growth}

    For convenience, we identify locations $(x,y)$ in the two-dimensional real plane $\RR^2$ with the complex numbers $x+iy \in \CC$. In particular, $n\in \NN$ minutiae locations in a fingerprint are then given by the complex numbers 
    $$z_j = x_j +i y_j\in \CC\,\quad j=1,\ldots,n\,.$$
    In our model, we assume that there are two linearly superimposed growth effects for a fingerprint, both of which originate from a central location which can be approximated by the center of all minutiae. In fact, w.l.o.g. we may assume in the following that all minutiae configurations are centered, i.e. $$\frac{1}{n}(z_1+\ldots+z_n)=0\,.$$
    Then the first growth effect, which is isotropic at rate $\lambda >0$, is modeled via $$ z_j \mapsto \lambda z_j\,.$$
    The second growth effect is anisotropic and occurs along an axis determined  by $w = e^{i\gamma}, \gamma \in [0,\pi)$ at a rate $\tau >0$, resulting in
    \begin{eqnarray*} z_j\mapsto z_j + \tau \langle z_j,w\rangle w &=& z_j + \tau\frac{e^{i\gamma}\overline{z_j} + e^{-i\gamma}z_j}{2}e^{i\gamma}\\
     &=& \frac{2+\tau}{2}z_j + \frac{\tau}{2} e^{2i\gamma}\overline{z_j} 
    \,.
    \end{eqnarray*}
    Here we have used the standard Euclidean inner product of $\RR^2 \cong \CC$, 
    $$\langle z,w\rangle =\RE(\bar z w) = xu + yv\mbox{ for }z=x+iy, w=u+iv\,.$$
    
    In consequence, matching a centered query fingerprint with $n$ minutiae locations $$Z': z'_1,\ldots,z'_n$$
    with a centered template fingerprint with $n$ minutiae locations $$Z: z_1,\ldots,z_n$$ can be achieved by minimizing the distance functional
    \begin{eqnarray}\nonumber\label{eq:dist-fcnal} F(Z,Z'; \gamma, \beta,\tau,\lambda) &:=&  \sum_{j=1}^n \big| z'_j - \lambda e^{i\beta}(z_j + \tau \langle z_j,w\rangle w)\big|^2\\
     &=& \sum_{j=1}^n \left| z'_j - \lambda e^{i\beta}\left(\frac{2+\tau}{2}z_j + \frac{\tau}{2} e^{2i\gamma}\overline{z_j}\right)\right|^2
    \end{eqnarray}
    over the parameter space $$ \theta=(e^{i2\gamma}, e^{i\beta},\tau,\lambda) \in (\SSS^1)^2 \times [0,\infty)^2=\Theta\,.$$
    Here, $\SSS^1 = \{z\in \CC: |z|=1\}$ is the unit circle and $\beta \in [-\pi,\pi)$ models a rotational angle between $Z$ and $Z'$.
    
    We remark that with this choice of parameters we have avoided the following ambiguity. If $\tau < 0$ were possible, then isotropic growth followed by anisotropic growth could be equally modeled by larger isotropic growth followed by anisotropic shrinkage.
    
\subsection{Estimating Isotropic and Uni-Directional Anisotropic Growth}   
    Every parameter $\theta^*=(e^{i2\gamma^*}, e^{i\beta^*},\tau^*,\lambda^*)$ minimizing \eqnref{eq:dist-fcnal} is a critical point of \eqnref{eq:dist-fcnal}. Let $a\in\CC$. Exploiting the properties
    \begin{eqnarray*}
    \frac{\partial}{\partial \beta} \left(|a  - e^{i\beta}a'|^2\right) &=& -\frac{\partial}{\partial \beta}  (ae^{-i\beta}\overline{a'}+ e^{i\beta} a'\bar a) ~=~ i(e^{-i\beta}a\overline{a'}- e^{i\beta} \bar a a')\,,\\
       \frac{\partial^2}{\partial \beta^2}\left( |a  - e^{i\beta}a'|^2\right) &=& e^{-i\beta}a\overline{a'}+ e^{i\beta} \bar a a'\,,\\
    \frac{\partial}{\partial \lambda} \left(|a  - \lambda a'|^2\right) &=&2\lambda |a'|^2 - (a\overline{a'}+\bar a a') \, ,
    \end{eqnarray*}
     yields that
    \begin{eqnarray*}
     \argmin_{e^{i\beta}} |a  - e^{i\beta}a'|^2 &=& \frac{a\overline{a'}}{|a\overline{a'}|}\,,\\
    \argmin_{\lambda >0} |a  - \lambda a'|^2 &=& \frac{1}{2} \frac{a\overline{a'}+\bar a a'}{|a'|^2}\,.
    \end{eqnarray*}
    We thus obtain
    $$\begin{array}{rclcl} 
     e^{2i\gamma^*} &=&  e^{-i\beta^*}\frac{\sum_{j=1}^n \left(z'_j -\lambda^*e^{i\beta^*} \frac{2+\tau^*}{2} z_j\right) z_j}{\left|\sum_{j=1}^n \left(z'_j -\lambda^*e^{i\beta^*} \frac{2+\tau^*}{2} z_j\right)\,z_j\right|}&=:&A(e^{i\beta^*},\tau^*,\lambda^*)
     \\
     e^{i\beta^*}  &=&  \frac{\sum_{j=1}^n z'_j \left(\frac{2+\tau^*}{2}\overline{z_j} + \frac{\tau^*}{2} e^{-2i\gamma^*}z_j\right)}{ \left|\sum_{j=1}^n z'_j \left(\frac{2+\tau^*}{2}\overline{z_j} + \frac{\tau^*}{2} e^{-2i\gamma^*}z_j\right)\right|}&=:&B(e^{2i\gamma^*},\tau^*,\lambda^*)
     \\
     \tau^*&=&\frac{2}{\lambda}\,\frac{\sum_{j=1}^n\RE\left(\left(z'_j -\lambda^*e^{i\beta^*}z_j\right) \left(e^{-2i\gamma^*}z_j + \overline{z_j}\right)\right)}{\left|\sum_{j=1}^n \left(e^{-2i\gamma^*}z_j + \overline{z_j}\right)\right|^2} &=:&C(e^{2i\gamma^*},e^{i\beta^*} , \lambda^*)
     \\
     \lambda^* &=& \frac{\sum_{j=1}^n\RE\left(\overline{z'}_j e^{i\beta^*}\big((2+\tau^*)z_j + \tau^*e^{2i\gamma^*}\overline{z_j}\big) \right)}{\left|\sum_{j=1}^n\left((2+\tau^*)z_j + \tau^*e^{2i\gamma^*}\overline{z_j}\right)\right|^2}&=:&D(e^{2i\gamma^*},e^{i\beta^*} , \tau^*)\,.
     \end{array}
     $$

    Numerical experiments show that the following algorithm usually converges quickly. 
    \begin{Algo}\label{mainalgo}
    With an error threshold $\epsilon >0$,
    \begin{description}
     \item[$k=0$:] initialize $e^{i\beta_0} := 1 =: e^{2i\gamma_0}$ as well as $\lambda_0:=1$, $\tau_0 :=0$
     \item[$k \to k+1$:] compute 
     \begin{eqnarray*}
      e^{i\beta_{k+1}}&:=& B(e^{2i\gamma_k},\tau_k,\lambda_k)\\
      \lambda_{k+1} &:=& D(e^{2i\gamma_k},e^{i\beta_{k+1}},\tau_k)\\
      e^{2i\gamma_{k+1}}&:=& A(e^{i\beta_{k+1}},\tau_k,\lambda_{k+1})\\
      \tau_{k+1}&=& C(e^{2i\gamma_{k+1}},e^{i\beta_{k+1}},\lambda_{k+1})\,.      
     \end{eqnarray*}
      \item[break if] 
       $$ |e^{2i\gamma_{k+1}}-e^{2i\gamma_{k}}|^2 + |e^{i\beta_{k+1}}-e^{i\beta_{k}}|^2 + |\tau_{k+1}-\tau_k|^2 +|\lambda_{k+1}-\lambda_k|^2 < \epsilon\,.$$
    \end{description}
    \end{Algo}

    \begin{Rm} This algorithm is inspired by \emph{General Procrustes Analysis} by \cite{Gow}, cf. also \cite{DM98}. Setting $\tau \equiv 0$ and not estimating $\gamma$ is equivalent to \emph{Full Procrustes Analysis}, additionally setting $\lambda\equiv 1$ leads to \emph{Partial Procrustes Analysis}.\end{Rm} 
    
    The validation and sensitivity studies following each test in the following Section \ref{scn:tests} illustrate that Algorithm \ref{mainalgo} rather well identifies $\gamma$ and $\tau$. This can also be seen in Section \ref{scn:fidelity} where we also illustrate that Algorithm \ref{mainalgo} rather well recovers underlying variable growth with a tendency to overestimate $\tau$ and underestimate $\lambda$, where this latter effect strongly wanes with increased growth and anisotropy. 
    
\section{Testing For Anisotropic Growth}\label{scn:tests}

      Suppose we have fingerprints from $P$ individuals. For convenience we assume that for every individual $p=1,\ldots,P$, at a first time point, fingerprint $Z^{p}_0$ has been recorded and at other time points $k_p$ fingerprints have been recorded, denoted by $Z^{p}_1,\dots, Z^{p}_{k_p}.$  For $Z=Z^{p}_0$ and $Z'=Z^{p}_1,\dots, Z^{p}_{k_p}$, Algorithm \ref{mainalgo} provides us with estimates
      \begin{eqnarray}\label{tau-gamma-estimates}
       \hat{\tau} = \left( \hat{\tau}_{k}^p \right)_{\substack{1 \leq p \leq P \\ 1 \leq k \leq k_p}}&\mbox{and}& \hat{\gamma} = \left( \hat{\gamma}_{k}^p \right)_{\substack{1 \leq p \leq P \\ 1 \leq k \leq k_p}},
    \end{eqnarray}
      of the parameters responsible for possible anisotropic growth. Notably, the number $n_k^p$ of minutiae common to $Z_0^p$ and $Z_k^p$ can be different from the number $n_{k'}^p$ of minutiae common to $Z_0^p$ and $Z_{k'}^p$ for $k\neq k'$. Table  \ref{table:TestOverview} gives an overview over different tests introduced in the following in order to analyze the distribution of these estimates.
      Our testing scheme naturally generalizes to situations with several impressions for the first time point. 
       
       \subsection{Validation and Simulation Schemes}\label{scn:validation-simulation-schemes}
       For the validations and simulations below
       we have  selected $P=8$ individuals beginning with the second individual from the FVC2002 DB2 database (\cite{FVC2002}). Here, for every individual $p$, $k_p+1=8$ impressions have been taken in subsequent sessions and we have manually marked and matched all corresponding minutiae patterns for $p=1,\ldots,8$. 
  	  \begin{enumerate}
      \item[(i)] For every individual $p=1,\ldots,8$ the first fingerprint is semi-manually aligned using the (extended) quadratic differential tool from \cite{HHM08,GottschlichTamsHuckemann2017perfect} such that the nearly parallel friction ridges above the crease are oriented vertically, so that the positive horizontal axis points into the distal direction, see Figure \ref{fig:ausrichtengraphik}. This gives the aligned point patterns $(Z_0^p)_{1\leq p\leq 8}$.      

      \begin{figure}[ht!]
      	\centering
      	\includegraphics[width=0.8\linewidth]{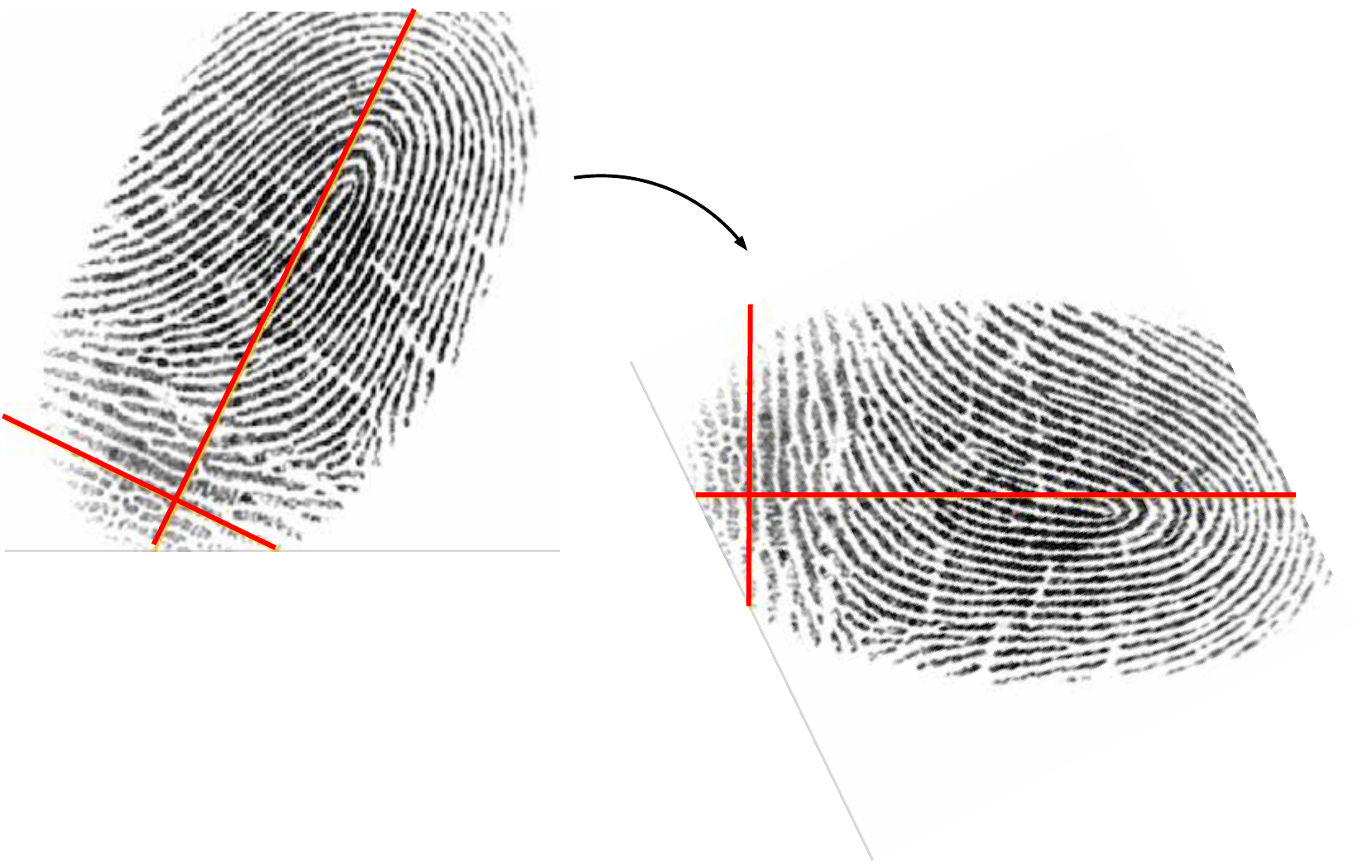}
      	\caption{\it Semi-automated alignment (FVC2002 DB2 Finger 7 Print 4) of crease lines along the vertical axis, using the (extended) quadratic differential tool from \cite{HHM08,GottschlichTamsHuckemann2017perfect}, such that the horizontal axis coincides with the distal axis.}
      	\label{fig:ausrichtengraphik}
      \end{figure}
       
      \item[(ii)] As described above, we apply Algorithm \ref{mainalgo} to the aligned point pattern $Z^p_0$ and the corresponding minutiae point patterns $Z^p_1,\ldots,Z_{7}^p$ extracted from the other impressions for $p=1,\ldots,8$, giving $\hat \tau$ and $\hat \gamma$ as in \eqnref{tau-gamma-estimates} which are used to validate the tests for anisotropy in Sections \ref{Test_graphical} and \ref{Test_tau}.
      
            \item[(iii)] Next, in order to simulate anisotropic growth of a given rate $\tau$ and orientation $e^{i\gamma}=w$ we align each $Z^p_1,\ldots,Z_{7}^p$ to $Z_0^p$  using Partial Procrustes Analysis giving rotations $\beta_1^p,\ldots,\beta_7^p$, 
             and set
      \[ \widetilde{Z}_k^p=e^{i\beta^p_k}Z_k^p+\tau ww^Te^{i\beta^p_k}Z_k^p \quad k=1,\ldots, 7, p=1,\ldots,8\,.\] 
      Applying Algorithm \ref{mainalgo} to $Z_0^p$ and $\widetilde{Z}_1^p,\ldots,\widetilde{Z}_7^p$ ($p=1,\ldots,8$), the corresponding estimates $\hat \tau$ and $\hat \gamma$ are then subjected to tests in Sections \ref{Test_tau}, \ref{Test_distal_param} and \ref{Test_distal_nonparam}.
      \item[(iv)] In fact, for the latter two tests we require a \emph{fingerprint alignment precision parameter} $\eta$, reflecting how accurately we can detect the distal axis in a given fingerprint. To this end, we align all other fingerprints as in (i) to obtain the point patterns $\left( {Z'}_{k}^p \right)_{\substack{1 \leq p \leq 8 \\ 1 \leq k \leq 7}}$ and apply once again Algorithm \ref{mainalgo} where we record only the nuisance parameters $\hat\beta = \left( \hat{\beta}_{k}^p \right)_{\substack{1 \leq p \leq 8 \\ 1 \leq k \leq 7}}$. Figure \ref{fig:Boxplot_beta.png} shows the distribution of $\hat\beta$ and we choose for the fingerprint alignment precision parameter approximately the maximal quartile 
      \begin{equation}
       \label{epsilon:def}
        \eta=0.075\,
      \end{equation}
      which corresponds to approximately $4$ degrees.

      \begin{figure}[h!]
      	\centering
      	\includegraphics[width=0.6\linewidth]{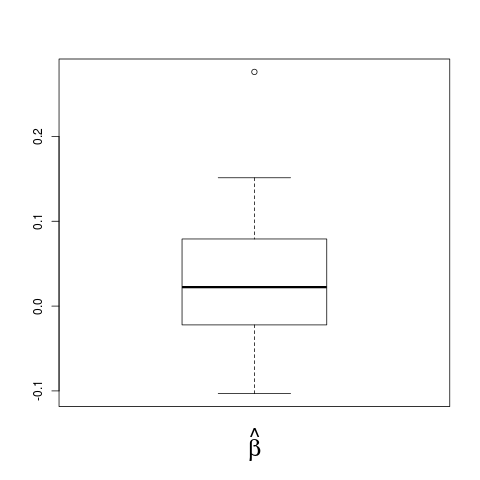}
      	\caption{\it Fingerprint alignment precision: Boxplot of $56$ rotation angles $\hat\beta $ 
      	obtained from Algorithm \ref{mainalgo} applied to semi-automatically aligned fingers, cf. Figure \ref{fig:ausrichtengraphik}.}
      	\label{fig:Boxplot_beta.png}
      \end{figure}

      \item[(v)] We simulate realizations of homogeneous Poisson point processes in rectangles of varying aspect ratios and distort them by i.i.d.~Gaussian noise of varying variances, truncated at $5$ standard deviations. In the following model the minutiae of $Z_k^p$ result from the ones of $Z_0^p$ ($k=1,\ldots,k_p$, $p=1,\ldots,P$) by isotropic growth and truncated Gaussian perturbation,
            \begin{equation}
      \label{H_0:eq}
      z_{kj}^p = {\lambda}_k^p e^{i\beta_k^p} (z_{0j}^p + \epsilon_{kj}^p),~\epsilon_{kj}^p \iid \cN_{\rm truncated}(0,\sigma^2I_2), 
      ~j=1,\ldots,n^p_k\,.
      \end{equation}
       Of course this model, due to correlations caused by limited skin elasticity, is not realistic for single fingerprints. It is much more realistic, however, if we average over many imprints. Notably, in order to ensure that centeredness of $Z_0^p$ implies centeredness of $Z_k^p$ we would have to subtract the $\hat{\lambda}_k^p$-fold of $\bar\epsilon = \frac{1}{n^p_k}\sum_{j=1}^{n^p_k}\epsilon_j$. In the asymptotic scenario considered, however, we can neglect this effect, as it is of order $O_p\left(1/\sqrt{n^p_k}\right)$. 
       
      When we set ${\lambda}_k^p =1$, mimic the noise level $\sigma$ due to distortions in real fingerprints and mimic the aspect ratios of real minutiae point patterns, we obtain a \emph{reference distribution} for a test for $\tau$.  In order to verify that our algorithm is not affected by anisotropy present in  point patterns, we have also considered rectangles of varying aspect ratios and varying isotropic growth. This validation is discussed in more detail in the following Section \ref{Test_graphical}.     
      
      \item[(vi)] Finally, in order to assess validity and sensitivity of our entire tool chain in Section \ref{scn:fidelity}, also for variable growth, we simulate growth as in (iii) also for $\tau$ and $\lambda$ following truncated Gaussians.
     \end{enumerate}

      \subsection{Test for $\hat \gamma$}\label{Test_graphical}

      Recall that the angle $\gamma \in [0,\pi)$ gives the orientation of anisotropic growth if $\tau>0$ and in case of no anisotropy ($\tau=0$), although $\gamma$ is meaningless, Algorithm \ref{mainalgo}, by design, returns estimates $\hat \gamma$. In that case, one may be tempted to expect a uniform distribution on $[0,\pi)$ and apply a standard test for uniformity of $2\gamma$ on $[0,2\pi)$ such as the \emph{Rayleigh Test}, cf. \cite[pp. 94--95]{MJ00}. Here, one considers the \emph{resultant length}, which is the length of the mean direction,
      \begin{eqnarray}\label{resultant}
       \hat R &=& \frac{1}{\sum_{p=1}^Pk_p}\left|\sum_{p=1}^P\sum_{k=1}^{k_p} e^{2i\hat\gamma_j^p}\right|\,,
       \end{eqnarray}
      and, under uniformity, the statistic $ 2\sum_{p=1}^P k_p  \hat R^2$ follows asymptotically for $P\to \infty$ a $\chi^2_2$ distribution. Thus, with the $\alpha$-quantile $\chi^2_{2,\alpha}$ of a $\chi^2_2$ we obtain the first test. 

     \begin{Test}\label{test:Rayleigh} Reject $H_0$ that there is no anisotropic growth effect with confidence $1-\alpha$ $(\alpha \in [0,1])$ if
      $$2\sum_{p=1}^P k_p  \hat R^2 > \chi^2_{2,1-\alpha}\,.$$
     \end{Test}

     \begin{figure}[ht!]
      	\centering
      	\includegraphics[width=0.5\linewidth]{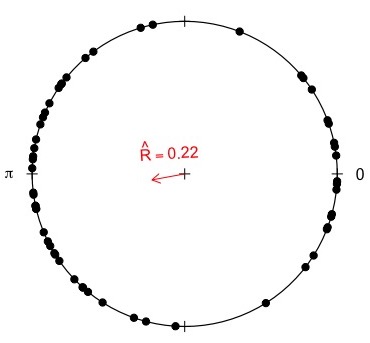}
	\includegraphics[width=0.45\linewidth]{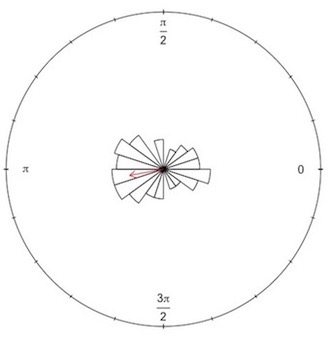}
      	\caption{\it Left: circular plot of $56$ estimates of $e^{2i\gamma}$ by Algorithm \ref{mainalgo} and their resultant length $R$. Right: their suitably binned rose diagram (cf. \cite{MJ00}). In both displays, the red arrow points toward the extrinsic mean from \eqnref{eq:extrinsic-mean}.
      	}
	\label{fig:31graphicaltestresults}
	\end{figure}

      \begin{figure}[ht!]
  		\centering
  		\includegraphics[width=0.7\linewidth]{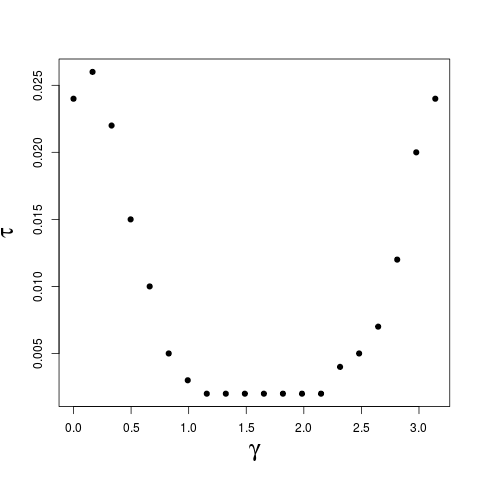}
  		\caption{\it The smallest $\tau$ such that  
  		Test \ref{test:Rayleigh}  
  		significantly detects anisotropy over varying orientations $\gamma$ of growth. }
  		\label{fig:graphicgammatauRquer}
  	\end{figure}

  	\paragraph{Validation.}
    As described in (ii) of Section \ref{scn:validation-simulation-schemes}, for the data at hand, we have 
 	\[2*56*0.2188^2 = 5.36\ngtr \chi^2_{2,0.95}= 5.99\,,\]
      so, as expected, we cannot reject uniformity with $95\,\%$ confidence. In Figure \ref{fig:31graphicaltestresults} we display the doubled estimated orientations. Although we can clearly see a bimodal pattern, with a clear preference for distortions along the  medial-lateral axis ($2\gamma = \pi$) and smaller preference along the proximal-distal axis ($\gamma =0$), as we have tested, this is not significant. One may be inclined to attribute this effect to the algorithm used and the anisotropy inherent in most minutiae patterns, namely that they extend more along the proximal-distal axis than along the medial-lateral axis. Simulations with larger anisotropy in the point patterns, as described in (v) of  Section \ref{scn:validation-simulation-schemes}, however, endorse uniformity of $2\gamma$ under $H_0$.
      
      More precisely, as described in (v) of Section \ref{scn:validation-simulation-schemes} we mimic the number of fingers (8) and impressions (8), minutia intensity (22) and noise's standard deviation (7). For rectangles with horizontal lengths $60,100,200,400,600$ and vertical lengths $200$ with isotropic growth rates $\lambda = 1,2,3,4$ we have simulated a total of $1'000$ point patterns and Test \ref{test:Rayleigh} rejected isotropic growth $61$ times at level $\alpha =0.05$ and $15$ times at level $\alpha = 0.01$, i.e. the test roughly holds the size.

      \paragraph{Sensitivity.}	 For every $\gamma\in\left\{ 0,\frac{1}{20}\pi,\dots,\frac{19}{20}\pi\right\}$ and $\tau >0$ we have simulated data as in (iii) from Section \ref{scn:validation-simulation-schemes}, and we have computed the minimal $\tau = \tau(\gamma)$ for which Test \ref{test:Rayleigh} 
 	 rejects $H_0$ that there is no anisotropic growth with confidence $0.95$.  
 	 The corresponding values are plotted in Figure \ref{fig:graphicgammatauRquer} and we conclude that, for the data at hand, Test \ref{test:Rayleigh} significantly detects anisotropy if $\tau \geq 0.026$. Because under $H_0$ (growth is isotropic), for this dataset, $\hat \gamma$ is biased towards $\pi/2$ (corresponding to the medial-lateral axis), detection of anisotropy directed along the proximal-distal axis ($\gamma = 0,\pi$) is slightly more difficult than along the medial-lateral axis ($\gamma=\pm\pi/2$).

      \subsection{Testing for $\hat{\tau}$}\label{Test_tau}
      
      While under $H_0$ (no anisotropy), $\hat \gamma$ is rather uniformly distributed, the distribution of $\hat \tau$ is less clear. Since $\hat \lambda$ and $\hat \lambda(1+\hat \tau)$ can be viewed as the smallest and the largest distortions, respectively, one may want to model their distribution by the likelihood ratio statistic
            \begin{equation}
      \label{sphericity:stats}  \rho_k^p = 2n^p_k \log \frac{\hat a_k^p}{\hat g_k^p}\,,\end{equation}
      which, under Gaussianity, asymptotically ($n^p_k\to \infty$) follows a $\chi^2_{2}$-distribution,
      where $\hat a_k^p$ denotes the arithmetic mean and $\hat g_k^p$ the geometric mean of the two eigenvalues of the empirical covariance matrix (cf. \citet[p. 124 and p. 134]{MardiaKentBibby1980}). This is not true, however, for the estimates produced by Algorithm \ref{mainalgo}.  Figure \ref{fig:qqrho-chisq} shows that under $H_0$ the r.h.s of  \eqnref{sphericity:stats} deviates in size and shape considerably from $\chi^2_2$. 
      
  	\begin{figure}[hb!]
  		\centering
  		\includegraphics[width=0.7\linewidth]{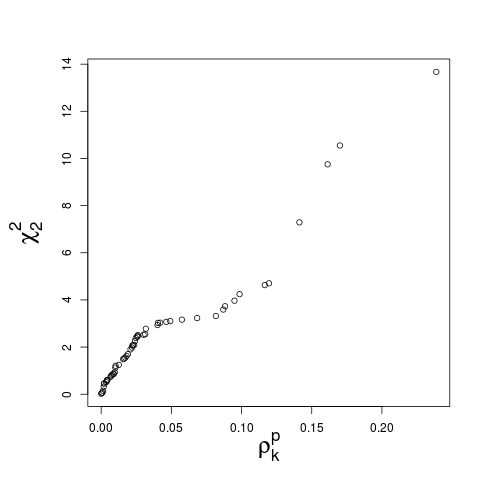}
  		\caption{\it QQ-plot of the likelihood ratio statistic from \eqnref{sphericity:stats} vs. $\chi^2_2$.}
  		\label{fig:qqrho-chisq}
  	\end{figure}
      
      For this reason we simulate $\tau$ from Poisson deviates, as in (v) of Section \ref{scn:validation-simulation-schemes} where we mimic the number of fingers ($8$) and impressions ($8$), of minutiae (their rounded mean is $22$ which we take for the minutiae intensity), aspect ratios of point patterns (we take their mean maximal extensions of $[-95,95]\times [160,160]$) and noise due to distortions (we choose a standard deviation of $7$ which is the rounded root mean square of empirical variances) from the $64$ fingerprints at hand. The  
      \begin{align}\label{tau-reference:eq} \widetilde{\tau} &= \left( \widetilde{\tau}_{k}^p \right)_{\substack{1 \leq p \leq P \\ 1 \leq k \leq k_p}}\end{align}
      thus obtained serve as the \emph{reference sample}.     
      
      \begin{Test} 
      \label{testfortau} Reject $H_0$ that there is no anisotropic growth effect with confidence $1-\alpha$ $(\alpha \in [0,1])$ if 
      a test for equality of distributions of samples, e.g. the \emph{Kolmogorov-Smirnov Goodness of Fit Test}, rejects the equality of distributions of $\hat \tau$ and $\widetilde{\tau}$.
      \end{Test}

      \paragraph{Validation.} For the data at hand, the above test roughly holds the size. We have simulated $1,\!000$ reference samples $\widetilde{\tau}$, and for the confidence level $\alpha = 0.05$ we have observed a size of $0.952$. For the confidence level $\alpha = 0.01$ the size was $0.993$.
      
  	\begin{figure}[ht!]
  		\centering
  		\includegraphics[width=0.7\linewidth]{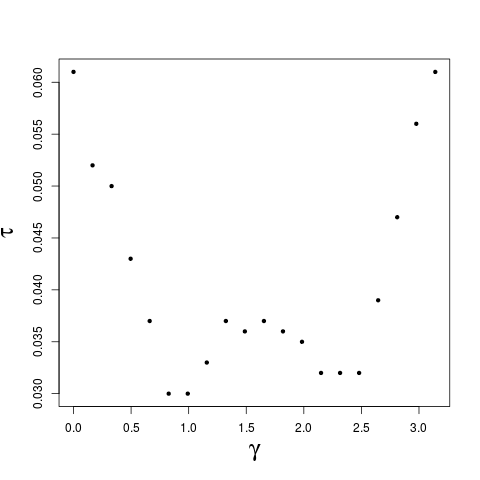}
  		\caption{\it The smallest $\tau$ such that  
  		Test \ref{testfortau}  
  		significantly detects anisotropy over varying orientations $\gamma$ of growth. }
  		\label{fig:graphicgammatau}
  	\end{figure}
 	 \paragraph{Sensitivity.}
 	 As in the previous section, for every $\gamma\in\left\{ 0,\frac{1}{20}\pi,\dots,\frac{19}{20}\pi\right\}$ and $\tau >0$ we have simulated a single reference sample $\widetilde{\tau}$ from Equation (\ref{tau-reference:eq}), simulated data as in (iii) from Section \ref{scn:validation-simulation-schemes}, and we have computed the minimal $\tau = \tau(\gamma)$ for which 
 	 the above test 
 	 rejects $H_0$ that there is no anisotropic growth with confidence $0.95$.  
 	 The corresponding values are plotted in Figure \ref{fig:graphicgammatau} and we conclude that we can significantly detect anisotropy if $\tau \geq 0.061$. Notably, this number is random as it depends on the specific reference sample. Moreover, with the specific reference sample, as in  Figure \ref{fig:graphicgammatauRquer}, detection of anisotropy directed along the proximal-distal axis ($\gamma =  0,\pi$) is slightly more difficult than along the medial-lateral axis ($\gamma = \pm \pi/2$). 
%
      \subsection{Testing Jointly for $e^{2i\hat{\gamma}}$ and $\hat{\tau}$} \label{Test_jointly}
      
      As we have seen in the previous section, the marginal distribution of $\hat \tau$ and thus of $\hat \rho$ obviously depends on the algorithm used and cannot be easily assessed over distributions of eigenvalues, say \citeegp{silverstein1995empirical,johnstone2001distribution,nadler2011distribution,wei2012exact}. 
      In consequence, we propose a test that is similar in spirit to the previous Test \ref{testfortau}. 
   
      Marking a large database and / or employing resampling methods by leaving out some minutiae, say, one can arrive at empirical confidence regions for 
      $$\hat\theta = \left(\hat R,\sum_{p=1}^P\sum_{k=1}^{k_p}\rho_k^p\right)\,$$
      with $\rho_k^p$ from  \eqnref{sphericity:stats}. To this end, simulate or obtain from a large database
      $$ \theta^{(i)}=\left(\hat R^{(i)}, \sum_{p=1}^P\sum_{k=1}^{k_p}\rho^{p,(i)}_k\right)\,, $$ 
      with common $\sum_{p=1}^P k_p = PK$ ($k_p = K$ for all $p=1,\ldots,P$), say, where $i = 1,\dots, N$ and $N$ is large (e.g. $N=1000$), denoting the number of replicates. Then for a given confidence level $\alpha\in [0,1]$ compute
       joint confidence rectangles 
      $$C_{1-\alpha} = \left\{\theta \in \RR^2: \big|\langle \theta-{\eta}_0,e_1 \rangle\big| \leq \lambda_1 c_{1-\alpha},~\big|\langle \theta-\theta_0,e_2 \rangle\big| \leq \lambda_2 c_{1-\alpha}\right\}$$
      parallel to the eigenvectors $e_1,e_2$ of the empirical covariance matrix with corresponding eigenvalues $\lambda_1,\lambda_2 >0$ of the $\theta^{(i)}$, centered at their sample mean $\theta_0$, by choosing suitable $c_{1-\alpha} >0$ such that
      $$ \sharp\{\theta^{(i)}\in C_{1-\alpha}: i=1,\ldots,N\} = \left\lfloor (1-\alpha) N\right\rfloor\,.$$
      
      \begin{Test}\label{test:jointly} Reject $H_0$ that there is no anisotropic growth effect with confidence $1-\alpha$ ($\alpha \in [0,1]$) if
       $$\hat\theta = \left(\hat R,\sum_{p=1}^P\sum_{k=1}^{k_p}\rho_k^p\right)\not\in C_{1-\alpha}\,. $$ 
      \end{Test}

      Implementation, validation and sensitivity study of this test is left for future work with access to larger marked databases. In fact, due to joint testing, we expect that Test \ref{test:jointly} has a higher power than the previous Tests \ref{Test_graphical} and \ref{Test_tau}, and can be used for smaller sample sizes with $P < 8$ and $k_p < 8$. 
      
      \subsection{Testing for Distal Growth: Parametric}\label{Test_distal_param}
      If we reject isotropic growth, we would like to test whether growth prefers a particular axis. For the task at hand, we consider only the proximal-distal axis.
      We note that the corresponding tests developed in this and in the subsequent section have a straightforward generalization to test for any given axis and we use these generalized tests in the validation and sensitivity studies below.
      
      In order to detect proximal-distal growth we want to reject the null hypothesis that growth occurs along any other axis. We have to take into account, though, that we can only identify the  proximal-distal axis, which we aligned with the horizontal axis, with a certain precision $\eta\in [0,\pi/2)$. This alignment precision parameter $\eta$ has to be estimated separately and to this end, in Section \ref{scn:validation-simulation-schemes} under (iv) in \eqnref{epsilon:def}, we have proposed a method. In consequence, setting for convenience $\mu = 2\gamma$, with $\gamma \in [-\pi/2,\pi/2)$, we consider the following hypotheses
           \begin{eqnarray} \label{test:distal-growth} 
       H'_0: |\mu| \geq \epsilon &\mbox{ vs. }&H'_1: |\mu| < \epsilon\,,
       \end{eqnarray}
      where, as desired, $H'_0$ reflects any axis other than the horizontal, under a given accuracy $\epsilon = 2\eta$. In this section, for the hypotheses in (\ref{test:distal-growth}), we propose a parametric test, and in the next section a nonparametric test.

      For the parametric test, for the null hypothesis we use an analog of the normal distribution for the circle, namely the \emph{von Mises distribution} with respect to the half circle density $d\gamma/\pi$
      $$ \gamma \mapsto I_0(\kappa)^{-1} \,e^{\kappa \langle e^{2i\gamma},e^{i\mu}\rangle} = I_0(\kappa)^{-1} \,e^{\kappa \cos(2\gamma-\mu)} $$ 
      with central angle $\mu \in [-\pi,\pi)$, concentration parameter $\kappa >0$ and integration constant
      $$ I_0(\kappa) = \int_{-\pi/2}^{\pi/2} e^{\kappa \cos 2\gamma}\,\frac{d\gamma}{\pi} = \int_{-\pi}^{\pi} e^{\kappa \cos t}\,\frac{dt}{2\pi}\,,$$
      which is given by the modified Bessel function of the first kind of order zero (e.g. \citet[p. 36]{MJ00}). 
      It is easy to see that the MLE for $\mu$ is given by the \emph{extrinsic mean} 
      \begin{eqnarray}\label{eq:extrinsic-mean}
      \overbar{\mu} &=& \Arg \left( \sum_{p=1}^P \sum_{k=1}^{K} e^{2i\hat{\gamma}_k^p}\right)\,,\end{eqnarray}
      where we conveniently choose the argument in $[-\pi,\pi)$. As detailed in \citet[pp. 70, 124]{MJ00}, $\overbar{\mu}$ conditioned on $\kappa$ and a resultant length $R$
      follows a von Mises distribution with the same $\mu$ and concentration parameter $ \kappa  R$. Hence, denoting with $\hat R$ the sample resultant length of \eqnref{resultant}, for the following test, we first estimate $\kappa$ by its MLE or approximate  \begin{eqnarray*}
      \hat \kappa = \begin{cases}
       		2 \hat{R} + \hat{R}^3 + \frac{5}{6} \hat{R}^5, & \text{if } \hat{R} < 0.53, \\
       		-0.4 + 1.39 \hat{R} + 0.43 \frac{1}{1-\hat{R}}, & \text{if } 0.53 \leq \hat{R} < 0.85, \\
          	\frac{1}{2(1-\hat{R})}, & \text{if } \hat{R} \geq 0.85, \\ 
      \end{cases}
      \end{eqnarray*}
      as discussed in \citet[Section 5.3.1]{MJ00} and then  numerically compute suitable 
      $0<\delta = \delta_{1-\alpha}<\epsilon$ such that
      \begin{eqnarray*}\label{vM-anisotropy:ori-problem} \nonumber
      1-\alpha &=&\Prb\left\{|\overbar{\mu}| > \delta :   \frac{1}{\sum_{p=1}^P k_p} \left|\sum_{p=1}^P\sum_{k=1}^{k_p} e^{2i\hat{\gamma}_k^p}\right| = \hat R, \mu = \epsilon\right\}
      \\
      &=& \frac{1}{I_0(\hat \kappa \hat R)}\int_{\delta}^{2\pi-\delta} e^{\hat \kappa \hat R \cos (\overbar{\mu}-\epsilon) }\,\frac{d\overbar{\mu}}{2\pi} \\
      &=& \frac{1}{I_0(\hat \kappa \hat R)}\int_{\delta}^{2\pi-\delta} e^{\hat \kappa \hat R \cos (\overbar{\mu}+\epsilon) }\,\frac{d\overbar{\mu}}{2\pi} \\
      &=& \Prb\left\{|\overbar \mu| > \delta :   \frac{1}{\sum_{p=1}^P k_p} \left|\sum_{p=1}^P\sum_{k=1}^{k_p} e^{2i\hat{\gamma}_k^p}\right| = \hat R, \mu = -\epsilon\right\}\,.   \nonumber
      \end{eqnarray*}
      In consequence, we have the following test.
      
      \begin{Test}\label{vM-anisotropy:test} Reject $H'_0$ that growth occurs not along the distal axis  with confidence $1-\alpha$ ($\alpha \in [0,1]$) if $$|\overbar \mu| <\delta\,.$$
      \end{Test}

            In fact, this test keeps the level only for $|\mu| = |\pm\epsilon|$. For $N= \sum_{p=1}^P n_k^p\to \infty$, due to asymptotic consistency of the MLE (e.g. \citet[p. 86]{MJ00}), under $|\mu| > |\epsilon|$ we have asymptotically that the rejection probability tends to zero and for $|\mu| <\epsilon$ the power tends asymptotically to one.

\subsection{Testing for Distal Growth: Nonparametric} \label{Test_distal_nonparam}
      For the nonparametric test 
      we first determine the extrinsic mean $\overbar \mu$ from \eqnref{eq:extrinsic-mean} and
      bootstrap the variance $V_{\overbar\mu^2}$ of its square. To this end, generate $B$ 
      bootstrap samples of the estimated gammas, compute the variance  $V^*_{|\overbar\mu|^2}$ of the $B$ squared extrinsic means and set
      $$ T = \frac{\overbar \mu^2 - \epsilon^2}{\sqrt{V^*_{\overbar \mu^2}}}\,,$$
      to 
      obtain the following nonparametric test. Here $\phi_{\alpha}$ is the $\alpha$ quantile of the standard normal distribution.
      \begin{Test}\label{nonparametric-anisotropy:test}
       Reject $H'_0$ that growth occurs not along the distal axis  with confidence $1-\alpha$ $(\alpha \in [0,1])$ if
       $$T < \phi_{\alpha}\,.$$
      \end{Test}
      
      \begin{Th}
       For $\alpha \in [0,1]$,  as $N= \sum_{p=1}^P n_k^p\to \infty$, Test \ref{nonparametric-anisotropy:test}  
       asymptotically keeps the level for  $|\mu| = |\epsilon|$.  Under $|\mu| > |\epsilon|$ we have asymptotically that the rejection probability tends to zero and for $|\mu| <\epsilon$ the power tends asymptotically to one. 
      \end{Th}

      \begin{proof}
       If $\mu = \epsilon$, we have asymptotic normality of the extrinsic sample mean 
       $$\sqrt{N}(\overbar \mu -\epsilon)\to \cN(0,\sigma^2)\,,$$
       with a suitable variance $\sigma^2>0$ guaranteed by \citet[Theorem 2]{HL96} or \citet[Theorem 3.1]{BP05}. Using the $\delta$-method, this translates at once to the squares 
       $$\sqrt{N}(\overbar \mu^2 -\epsilon^2)\to \cN(0,2\sigma^4)\,.$$
       The assertion for $\mu = \epsilon$  follows now because by 
       \citet[Corollary 1]{Cheng2015}, the variance can be consistently estimated by bootstrap samples. Similarly, obtain the assertion in case of $\mu = -\epsilon$. The assertions on the asymptotic level and the asymptotic power for $|\mu|\neq \epsilon$ follow at once from the consistency of the extrinsic mean (e.g. \cite{Z77} and \citet[Theorem 3.4]{BP03}, cf. also \citet[Theorem 4.1]{MPPPR07} for a similar argument).
      \end{proof}

      \begin{figure}[t!]
       \includegraphics[width=0.45\textwidth, clip=true, trim= 3cm 1cm 1cm 1cm]{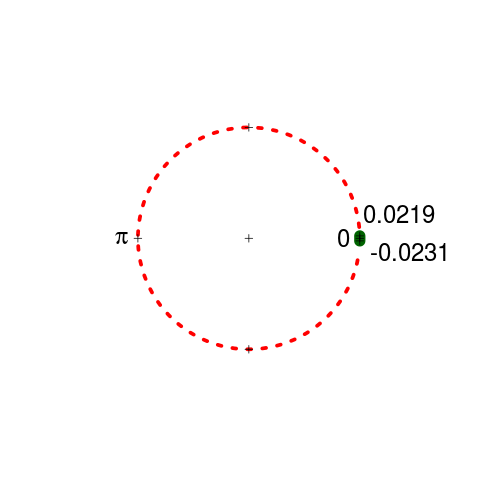} 
       \includegraphics[width=0.45\textwidth, clip=true, trim= 3cm 1cm 1cm 1cm]{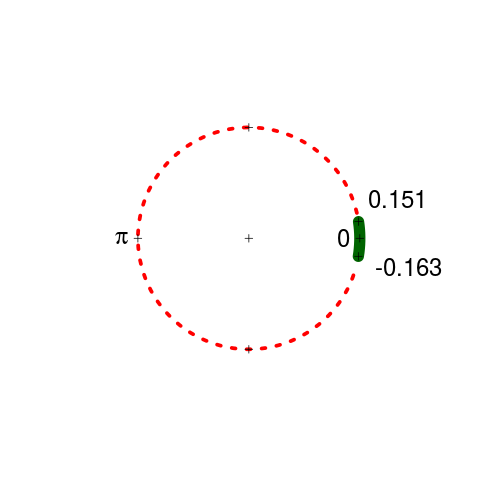}
       \caption{\it Simulating growth in directions $\gamma \in \left\{\frac{k \pi}{500}: k=0,\ldots,500\right\}$ of anisotropy $\tau=0.05$. Using a doubled alignment precision of $\epsilon = 0.15$, we apply Test \ref{vM-anisotropy:test} (left) and Test \ref{nonparametric-anisotropy:test} (right). Depicting the respective critical intervals for $2\gamma$, where distal growth is significantly detected, with solid green ($H'_0$ rejected) and the complement, where distal growth is not detected with dotted red ($H'_0$, that growth occurs non-distal, is not rejected). Notably, the angle zero corresponding to distal growth is always within the critical interval. }
       \label{fig:distal-growth-test}
      \end{figure}
\subsection{Validation and Sensitivity Study for Distal Growth Tests} 

      We have simulated growth of the marked fingerprints at hand as described in (iii) of Section \ref{scn:validation-simulation-schemes} with an anisotropy of $\tau = 0.05$ over varying directions $\gamma \in [0,\pi)$. Moreover we have applied Test \ref{vM-anisotropy:test} and Test \ref{nonparametric-anisotropy:test} with accuracy (doubled alignment precision) $\epsilon = 2\eta = 0.15$ for the doubled angles $2\gamma$, as estimated in \eqnref{epsilon:def}, and $B=100$. As Figure \ref{fig:distal-growth-test} illustrates, distal growth is significantly detected (i.e. $H'_0$, that growth occurs at a non-distal direction, is rejected) for both tests for the true value $\gamma=0$ which is contained in both critical intervals. This validates both tests. Moreover, for Test \ref{vM-anisotropy:test} the critical interval $[-0.0231, 0.0219]$ for $2\gamma$ is very narrow and well below alignment accuracy. In contrast, for Test \ref{nonparametric-anisotropy:test}, the critical interval $[-0.163, 0.151]$ for $2\gamma$ is of the order of alignment accuracy.
      
      
      Notably, the lack of symmetry of confidence intervals in Figure \ref{fig:distal-growth-test} is due to the data, as discussed in Section \ref{Test_tau}. Additionally in the right panel, the confidence interval is non-deterministic, due to bootstrapping.

              \begin{figure}[hb!]
      \centering
       \includegraphics[width=0.495\textwidth]{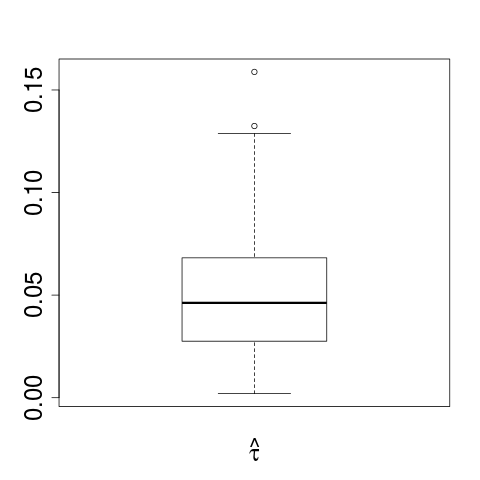}
       \includegraphics[width=0.495\textwidth]{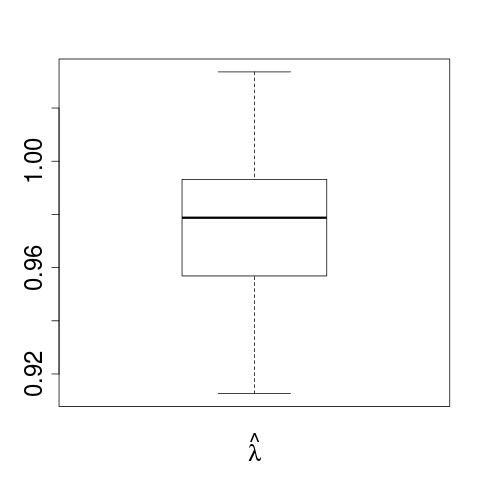}
       
              \caption{\it Distribution of $\hat \tau$ and $\hat\lambda$, estimated by Algorithm \ref{mainalgo}, in case of no growth at all: $\tau =0$ and $\lambda=1$.}\label{fig:simulationsvaryingtaulambda0}
              \end{figure}

 \section{Model vs. Algorithm Compatibility Study and Detection of Variable Growth} \label{scn:fidelity}

	Recall that, in order to validate Tests \ref{test:Rayleigh} and \ref{testfortau}, we applied Algorithm \ref{mainalgo} to the data at hand, which, as detailed before, features no growth. As elaborated in the introduction,  model and algorithm are not per se compatible. Rather, as is frequent in statistical shape analysis, we expect inconsistencies. 
	In Figure \ref{fig:simulationsvaryingtaulambda0} we report the algorithm's estimates $\hat \tau$ and $\hat\lambda$ for a model's ``true'' $\tau=0$ and ``true'' $\lambda=1$. Obviously, the model's $\tau$ is 
	overestimated 
	and $\lambda$ is 
	underestimated. 
	The tendency to overestimate a model's $\tau$ and underestimate a model's $\lambda$ is also visible in Figures \ref{fig:simulationsvaryingtaulambda1} and \ref{fig:simulationsvaryingtaulambda2} where we simulate non-zero growth. Notably, with increased growth this effect disappears quickly. 
	
	Next, for the data at hand, we simulate growth along the  proximal distal axis, but this time, for every individual imprint, we use different $\tau$ and $\lambda$, corresponding to imprints of children with varying unknown age increments.

      \begin{figure}[th!]
      \centering
       \includegraphics[width=0.5\textwidth]{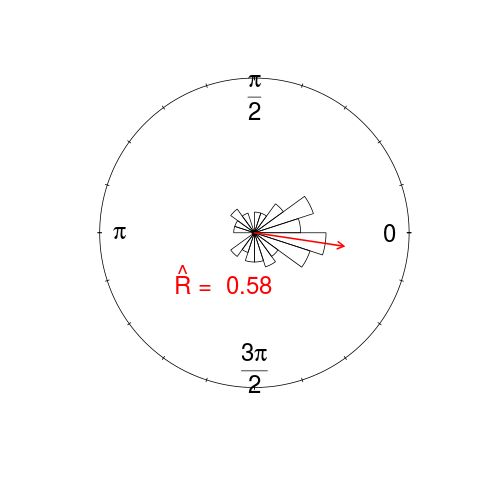}
       
       \includegraphics[width=0.495\textwidth]{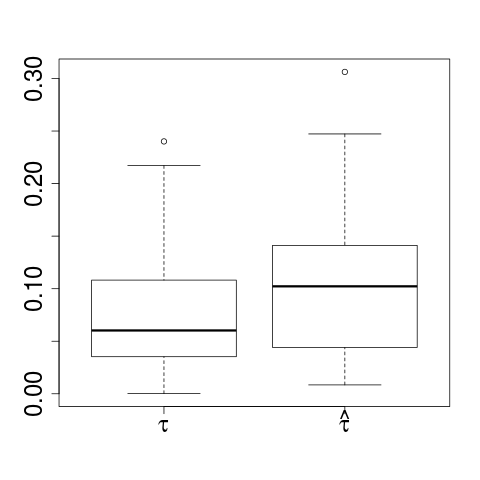}
       \includegraphics[width=0.495\textwidth]{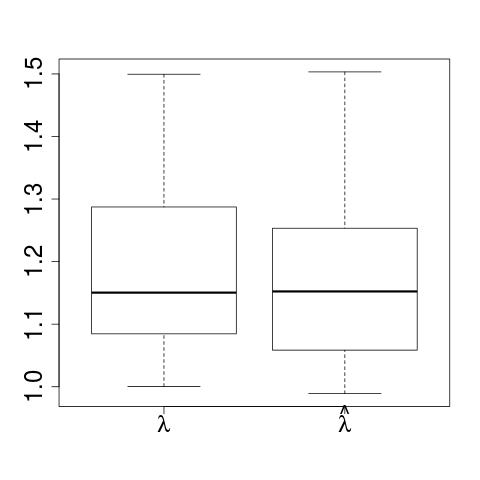}

              \caption{\it Distribution of $2\hat \gamma$ with extrinsic mean (direction of red arrow) and resultant length (top row) after simulating growth along the proximal-distal axis of varying $\tau$ and $\lambda$ and their estimates $\hat \tau, \hat \lambda$ (boxplots in the bottom row).}\label{fig:simulationsvaryingtaulambda1}
              \end{figure}

              The first simulation mimics very moderate growth, in the median an isotropic growth of $15\,\%$ (median of estimates is also $15\,\%$) with anisotropy of $0.06\,\%$ in the median (median of estimates is $10\,\%$), and a boxplot of the model's distribution with their estimates is depicted in the bottom row of Figure \ref{fig:simulationsvaryingtaulambda1}. The resulting estimates for $2\hat\gamma$ are clearly non-uniform (top row of Figure \ref{fig:simulationsvaryingtaulambda1}) and this is highly significantly detected  by Test \ref{test:Rayleigh}. Test \ref{testfortau} significantly detects growth.
              Because the model's true $\tau$ is rather diffuse, Test \ref{vM-anisotropy:test} significantly detects the distal axis only for an accuracy $\epsilon \geq 0.73$. Since Test \ref{nonparametric-anisotropy:test} is less conservative, it turns out that it significantly detects the distal axis already for alignment accuracy ($\epsilon = 0.15$). 
              While the boxplot of the estimated $\hat\lambda$ is only slightly below the one of the model's values $\lambda$, the corresponding boxplots for $\hat\tau$ and $\tau$ show that the estimates of the model's $\tau$ are still more spread out and too large.

      \begin{figure}[th!]
      \centering
       \includegraphics[width=0.5\textwidth]{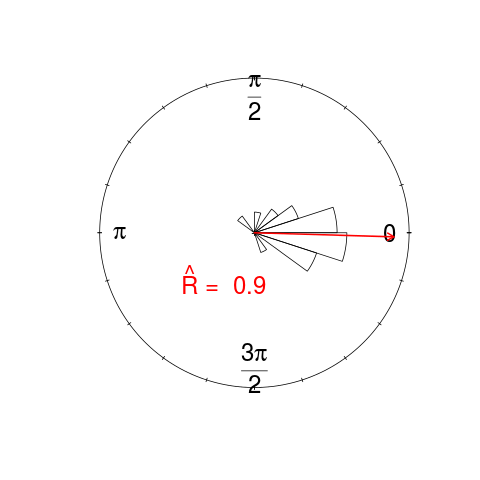}
       
       \includegraphics[width=0.495\textwidth]{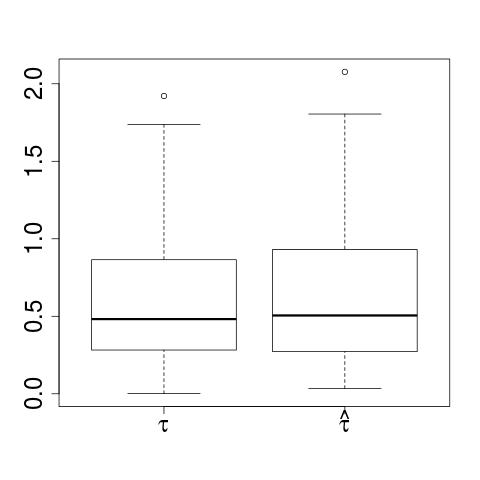}
       \includegraphics[width=0.495\textwidth]{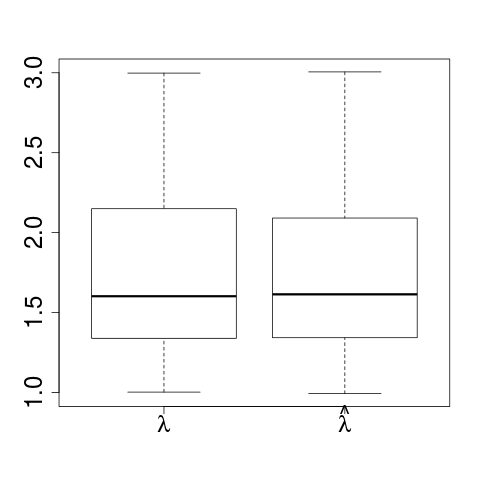}

              \caption{\it Notation as in Figure \ref{fig:simulationsvaryingtaulambda1}, but simulating much stronger growth.}\label{fig:simulationsvaryingtaulambda2}
              \end{figure}

             The second simulation mimics considerable growth, in the median an isotropic growth of $60 \,\%$ (median of estimates is $61\,\%$) with anisotropy of $0.45\,\%$ in the median (median of estimates is $51\,\%$), and a boxplot of the model's distribution with their estimates is depicted in the bottom row of Figure \ref{fig:simulationsvaryingtaulambda2}. The resulting estimates for $2\hat\gamma$ are even more clearly non-uniform (top row of Figure \ref{fig:simulationsvaryingtaulambda2}) and, again, this is highly significantly detected  by Test   \ref{test:Rayleigh}. Also, Test \ref{testfortau}  highly significantly detects growth.                  
              Because the model's $\tau$ are again rather diffuse, Test \ref{vM-anisotropy:test} significantly detects the distal axis only for roughly doubled alignment accuracy $\epsilon \geq 0.33$. Again Test \ref{nonparametric-anisotropy:test} is less conservative and it significantly detects the distal axis at alignment accuracy. 
              Now, both boxplots of the estimated $\hat \tau$ and $\hat\lambda$, respectively, are very similar to the ones of the model's values $\lambda$ and $\tau$, respectively. 

              \FloatBarrier

 \section{Outlook}


In a next step, beyond the scope of this paper, our methods can be applied to fingerprint growth data, (why such data are particularly difficult to obtain is highlighted in the recent study by the \cite{JRC2013}), to develop realistic fingerprint growth models that can be used in long-term health programs involving newborns, toddlers and children. Moreover, they can also be used against forgery of birth certificates, which are the weakest link in the identity chain, by including fingerprints which could then be projected into adulthood. Remarkably, according to an article in Le Parisien (19.12.2011, \cite{LeParisien2011}) between 0.5 and 1 million of a total of 6.5 million passports which are currently used in France are estimated to have been issued on the basis of forged breeder documents by the applicants.

\section*{Acknowledgements}

C. Gottschlich and S. Huckemann gratefully acknowledge the support of the Felix-Bernstein-Institute for Mathematical Statistics in the Biosciences, the Niedersachsen Vorab of the Volkswagen Foundation. All authors are indebted to support by the DFG RTG 2088.

\bibliography{FingerGrowth}

\begin{thebibliography}{}

\bibitem[\protect\citeauthoryear{Bhattacharya and Patrangenaru}{Bhattacharya
  and Patrangenaru}{2003}]{BP03}
Bhattacharya, R.~N. and V.~Patrangenaru (2003).
\newblock Large sample theory of intrinsic and extrinsic sample means on
  manifolds {I}.
\newblock {\em The Annals of Statistics\/}~{\em 31\/}(1), 1--29.

\bibitem[\protect\citeauthoryear{Bhattacharya and Patrangenaru}{Bhattacharya
  and Patrangenaru}{2005}]{BP05}
Bhattacharya, R.~N. and V.~Patrangenaru (2005).
\newblock Large sample theory of intrinsic and extrinsic sample means on
  manifolds {II}.
\newblock {\em The Annals of Statistics\/}~{\em 33\/}(3), 1225--1259.

\bibitem[\protect\citeauthoryear{Cheng}{Cheng}{2015}]{Cheng2015}
Cheng, G. (2015).
\newblock Moment consistency of the exchangeably weighted bootstrap for
  semiparametric m-estimation.
\newblock {\em Scandinavian Journal of Statistics\/}~{\em 42\/}(3), 665--684.

\bibitem[\protect\citeauthoryear{Devilliers, Allassonni{\`e}re, Trouv{\'e}, and
  Pennec}{Devilliers et~al.}{2017}]{DevilliersAllassonniereTrouvePenec2017}
Devilliers, L., S.~Allassonni{\`e}re, A.~Trouv{\'e}, and X.~Pennec (2017).
\newblock Inconsistency of template estimation by minimizing of the
  variance/pre-variance in the quotient space.
\newblock {\em Entropy\/}~{\em 19\/}(6), 288.

\bibitem[\protect\citeauthoryear{Dryden}{Dryden}{2009}]{Dshapes}
Dryden, I.~L. (2009).
\newblock R-package for the statistical analysis of shapes,
  http://www.maths.nott.ac.uk/$\sim$ild/shapes.

\bibitem[\protect\citeauthoryear{Dryden and Mardia}{Dryden and
  Mardia}{1998}]{DM98}
Dryden, I.~L. and K.~V. Mardia (1998).
\newblock {\em Statistical Shape Analysis}.
\newblock Chichester: Wiley.

\bibitem[\protect\citeauthoryear{Galton}{Galton}{1892}]{Galton1892}
Galton, S. (1892).
\newblock {\em Finger prints}.
\newblock Macmillan and Co.

\bibitem[\protect\citeauthoryear{Gottschlich}{Gottschlich}{2012}]{Gottschlich2012}
Gottschlich, C. (2012, April).
\newblock Curved-region-based ridge frequency estimation and curved {Gabor}
  filters for fingerprint image enhancement.
\newblock {\em {IEEE} Transactions on Image Processing\/}~{\em 21\/}(4),
  2220--2227.

\bibitem[\protect\citeauthoryear{Gottschlich}{Gottschlich}{2016}]{Gottschlich2016}
Gottschlich, C. (2016, February).
\newblock Convolution comparison pattern: An efficient local image descriptor
  for fingerprint liveness detection.
\newblock {\em PLoS ONE\/}~{\em 11\/}(2), e0148552.

\bibitem[\protect\citeauthoryear{Gottschlich, Hotz, Lorenz, Bernhardt,
  Hantschel, and Munk}{Gottschlich
  et~al.}{2011}]{GottschlichHotzLorenzBernhardtHantschelMunk2011}
Gottschlich, C., T.~Hotz, R.~Lorenz, S.~Bernhardt, M.~Hantschel, and A.~Munk
  (2011, September).
\newblock Modeling the growth of fingerprints improves matching for
  adolescents.
\newblock {\em {IEEE} Transactions on Information Forensics and
  Security\/}~{\em 6\/}(3), 1165--1169.

\bibitem[\protect\citeauthoryear{Gottschlich and Sch\"onlieb}{Gottschlich and
  Sch\"onlieb}{2012}]{GottschlichSchoenlieb2012}
Gottschlich, C. and C.-B. Sch\"onlieb (2012, June).
\newblock Oriented diffusion filtering for enhancing low-quality fingerprint
  images.
\newblock {\em IET Biometrics\/}~{\em 1\/}(2), 105--113.

\bibitem[\protect\citeauthoryear{Gottschlich, Tams, and Huckemann}{Gottschlich
  et~al.}{2017}]{GottschlichTamsHuckemann2017perfect}
Gottschlich, C., B.~Tams, and S.~Huckemann (2017).
\newblock Perfect fingerprint orientation fields by locally adaptive global
  models.
\newblock {\em IET Biometrics\/}~{\em 6\/}(3), 183--190.

\bibitem[\protect\citeauthoryear{Gower}{Gower}{1975}]{Gow}
Gower, J.~C. (1975).
\newblock Generalized {P}rocrustes analysis.
\newblock {\em Psychometrika\/}~{\em 40}, 33--51.

\bibitem[\protect\citeauthoryear{Hendriks and Landsman}{Hendriks and
  Landsman}{1996}]{HL96}
Hendriks, H. and Z.~Landsman (1996).
\newblock Asymptotic behaviour of sample mean location for manifolds.
\newblock {\em Statistics \& Probability Letters\/}~{\em 26}, 169--178.

\bibitem[\protect\citeauthoryear{Hotz, Gottschlich, Lorenz, Bernhardt,
  Hantschel, and Munk}{Hotz
  et~al.}{2011}]{HotzGottschlichLorenzBernhardtHantschelMunk2011}
Hotz, T., C.~Gottschlich, R.~Lorenz, S.~Bernhardt, M.~Hantschel, and A.~Munk
  (2011, September).
\newblock Statistical analyses of fingerprint growth.
\newblock In {\em Proc. BIOSIG}, Darmstadt, Germany, pp.\  11--20.

\bibitem[\protect\citeauthoryear{Huckemann}{Huckemann}{2011}]{H_Procrustes_10}
Huckemann, S. (2011).
\newblock Inference on 3{D} {P}rocrustes means: Tree boles growth,
  rank-deficient diffusion tensors and perturbation models.
\newblock {\em Scandinavian Journal of Statistics\/}~{\em 38\/}(3), 424--446.

\bibitem[\protect\citeauthoryear{Huckemann, Hotz, and Munk}{Huckemann
  et~al.}{2008}]{HHM08}
Huckemann, S., T.~Hotz, and A.~Munk (2008).
\newblock Global models for the orientation field of fingerprints: An approach
  based on quadratic differentials.
\newblock {\em IEEE Transactions on Pattern Analysis and Machine
  Intelligence\/}~{\em 30\/}(9), 1507--1519.

\bibitem[\protect\citeauthoryear{Jain, Arora, Best-Rowden, Cao, Sudhish, and
  Bhatnagar}{Jain et~al.}{2015}]{jain2015biometrics}
Jain, A.~K., S.~S. Arora, L.~Best-Rowden, K.~Cao, P.~S. Sudhish, and
  A.~Bhatnagar (2015).
\newblock Biometrics for child vaccination and welfare: Persistence of
  fingerprint recognition for infants and toddlers.
\newblock {\em arXiv preprint arXiv:1504.04651\/}.

\bibitem[\protect\citeauthoryear{Jain, Arora, Cao, Best-Rowden, and
  Bhatnagar}{Jain et~al.}{2017}]{Jain2017fingerprint}
Jain, A.~K., S.~S. Arora, K.~Cao, L.~Best-Rowden, and A.~Bhatnagar (2017).
\newblock Fingerprint recognition of young children.
\newblock {\em IEEE Transactions on Information Forensics and Security\/}~{\em
  12\/}(7), 1501--1514.

\bibitem[\protect\citeauthoryear{Jia, Cai, Gui, Hu, Lei, and Wang}{Jia
  et~al.}{2012}]{jia2012newborn}
Jia, W., H.-Y. Cai, J.~Gui, R.-X. Hu, Y.-K. Lei, and X.-F. Wang (2012).
\newblock Newborn footprint recognition using orientation feature.
\newblock {\em Neural Computing and Applications\/}~{\em 21\/}(8), 1855--1863.

\bibitem[\protect\citeauthoryear{Johnstone}{Johnstone}{2001}]{johnstone2001distribution}
Johnstone, I.~M. (2001).
\newblock On the distribution of the largest eigenvalue in principal components
  analysis.
\newblock {\em Annals of statistics\/}, 295--327.

\bibitem[\protect\citeauthoryear{{Joint Research Centre of the European
  Commission}}{{Joint Research Centre of the European
  Commission}}{2013}]{JRC2013}
{Joint Research Centre of the European Commission} (2013).
\newblock {\em Study on Fingerprint Recognition for Children, Final Report}.
\newblock Luxembourg: Publications Office of the European Union.

\bibitem[\protect\citeauthoryear{Kent and Mardia}{Kent and Mardia}{1997}]{KM97}
Kent, J.~T. and K.~V. Mardia (1997).
\newblock Consistency of {P}rocrustes estimators.
\newblock {\em Journal of the Royal Statistical Society, Series B\/}~{\em
  59\/}(1), 281--290.

\bibitem[\protect\citeauthoryear{Kotzerke, Arakala, Davis, Horadam, and
  {McVernon}}{Kotzerke et~al.}{2014}]{KotzerkeArakalaDavisHoradamMcVernon2014}
Kotzerke, J., A.~Arakala, S.~Davis, K.~Horadam, and J.~{McVernon} (2014,
  October).
\newblock Ballprints as an infant biometric: A first approach.
\newblock In {\em Proc. BIOMS}, Rome, Italy, pp.\  36--43.

\bibitem[\protect\citeauthoryear{K{\"u}cken and Newell}{K{\"u}cken and
  Newell}{2007}]{KuckenNewell2007}
K{\"u}cken, M. and A.~Newell (2007).
\newblock A model for fingerprint formation.
\newblock {\em EPL (Europhysics Letters)\/}~{\em 68\/}(1), 141.

\bibitem[\protect\citeauthoryear{Le}{Le}{1998}]{Le98}
Le, H. (1998).
\newblock {On the consistency of {P}rocrustean mean shapes.}
\newblock {\em Advances of Applied Probability (SGSA)\/}~{\em 30\/}(1), 53--63.

\bibitem[\protect\citeauthoryear{Lele}{Lele}{1993}]{Lele93}
Lele, S. (1993, July).
\newblock Euclidean distance matrix analysis ({EDMA}): estimation of mean form
  and mean form difference.
\newblock {\em Math. Geol.\/}~{\em 25\/}(5), 573--602.

\bibitem[\protect\citeauthoryear{Lemes, Segundo, Bellon, and Silva}{Lemes
  et~al.}{2014}]{LemesSecundoBellonSilva2014}
Lemes, R. d.~P., M.~P. Segundo, O.~R. Bellon, and L.~Silva (2014).
\newblock Dynamic pore filtering for keypoint detection applied to newborn
  authentication.
\newblock In {\em Pattern Recognition (ICPR), 2014 22nd International
  Conference on}, pp.\  1698--1703. IEEE.

\bibitem[\protect\citeauthoryear{Lemes, Bellon, Silva, and Jain}{Lemes
  et~al.}{2011}]{LemesBellonSilvaJain2011}
Lemes, R.~P., O.~R. Bellon, L.~Silva, and A.~K. Jain (2011).
\newblock Biometric recognition of newborns: Identification using palmprints.
\newblock In {\em Biometrics (IJCB), 2011 International Joint Conference on},
  pp.\  1--6. IEEE.

\bibitem[\protect\citeauthoryear{Maio, Maltoni, Capelli, Wayman, and Jain}{Maio
  et~al.}{2002}]{FVC2002}
Maio, D., D.~Maltoni, R.~Capelli, J.~L. Wayman, and A.~K. Jain (2002).
\newblock {FVC2002}: Second fingerprint verification competition.
\newblock In {\em Proc. ICPR}, pp.\  811--814.

\bibitem[\protect\citeauthoryear{Maltoni, Maio, Jain, and Prabhakar}{Maltoni
  et~al.}{2009}]{HandbookFingerprintRecognition2009}
Maltoni, D., D.~Maio, A.~K. Jain, and S.~Prabhakar (2009).
\newblock {\em Handbook of Fingerprint Recognition}.
\newblock London, U.K.: Springer.

\bibitem[\protect\citeauthoryear{Mardia and Jupp}{Mardia and Jupp}{2000}]{MJ00}
Mardia, K.~V. and P.~E. Jupp (2000).
\newblock {\em Directional Statistics}.
\newblock New York: Wiley.

\bibitem[\protect\citeauthoryear{Mardia, Kent, and Bibby}{Mardia
  et~al.}{1980}]{MardiaKentBibby1980}
Mardia, K.~V., J.~T. Kent, and J.~M. Bibby (1980).
\newblock {\em Multivariate analysis}.
\newblock Academic press.

\bibitem[\protect\citeauthoryear{Miolane, Holmes, and Pennec}{Miolane
  et~al.}{2017}]{MiolaneHolmesPennec2017}
Miolane, N., S.~Holmes, and X.~Pennec (2017).
\newblock Template shape estimation: correcting an asymptotic bias.
\newblock {\em SIAM Journal on Imaging Sciences\/}~{\em 10\/}(2), 808--844.

\bibitem[\protect\citeauthoryear{Munk, Paige, Pang, Patrangenaru, and
  Ruymgaart}{Munk et~al.}{2008}]{MPPPR07}
Munk, A., R.~Paige, J.~Pang, V.~Patrangenaru, and F.~Ruymgaart (2008).
\newblock The one- and multi-sample problem for functional data with
  application to projective shape analysis.
\newblock {\em Journal of Multivariate Analysis\/}~(99), 815--833.

\bibitem[\protect\citeauthoryear{Nadler}{Nadler}{2011}]{nadler2011distribution}
Nadler, B. (2011).
\newblock On the distribution of the ratio of the largest eigenvalue to the
  trace of a wishart matrix.
\newblock {\em Journal of Multivariate Analysis\/}~{\em 102\/}(2), 363--371.

\bibitem[\protect\citeauthoryear{Parisien}{Parisien}{2011}]{LeParisien2011}
Parisien, L. (2011, december).
\newblock Online available at
  \url{http://www.leparisien.fr/faits-divers/plus-de-10-des-passeports-biometriques-seraient-des-faux-19-12-2011-1775325.php}.

\bibitem[\protect\citeauthoryear{Sankaran, Vatsa, and Singh}{Sankaran
  et~al.}{2014}]{SankaranVatsaSingh2014}
Sankaran, A., M.~Vatsa, and R.~Singh (2014, September).
\newblock Latent fingerprint matching: A survey.
\newblock {\em IEEE Access\/}~{\em 2}, 982--1004.

\bibitem[\protect\citeauthoryear{Schumacher}{Schumacher}{2013}]{Schumacher2013}
Schumacher, G. (2013, September).
\newblock Fingerprint recognition for children.
\newblock {\em JRC Technical Reports\/}.

\bibitem[\protect\citeauthoryear{Silverstein and Bai}{Silverstein and
  Bai}{1995}]{silverstein1995empirical}
Silverstein, J.~W. and Z.~Bai (1995).
\newblock On the empirical distribution of eigenvalues of a class of large
  dimensional random matrices.
\newblock {\em Journal of Multivariate analysis\/}~{\em 54\/}(2), 175--192.

\bibitem[\protect\citeauthoryear{{SonLa~Study~Group}}{{SonLa~Study~Group}}{2007}]{Sonla2007using}
{SonLa~Study~Group} (2007).
\newblock Using a fingerprint recognition system in a vaccine trial to avoid
  misclassification.
\newblock {\em Bulletin of the World Health Organization\/}~{\em 85\/}(1), 64.

\bibitem[\protect\citeauthoryear{Thai, Huckemann, and Gottschlich}{Thai
  et~al.}{2016}]{ThaiHuckemannGottschlich2016}
Thai, D., S.~Huckemann, and C.~Gottschlich (2016, May).
\newblock Filter design and performance evaluation for fingerprint image
  segmentation.
\newblock {\em PLoS ONE\/}~{\em 11\/}(5), e0154160.

\bibitem[\protect\citeauthoryear{Wei, Tirkkonen, Dharmawansa, and McKay}{Wei
  et~al.}{2012}]{wei2012exact}
Wei, L., O.~Tirkkonen, P.~Dharmawansa, and M.~McKay (2012).
\newblock On the exact distribution of the scaled largest eigenvalue.
\newblock In {\em Communications (ICC), 2012 IEEE International Conference on},
  pp.\  2422--2426. IEEE.

\bibitem[\protect\citeauthoryear{Ziezold}{Ziezold}{1977}]{Z77}
Ziezold, H. (1977).
\newblock Expected figures and a strong law of large numbers for random
  elements in quasi-metric spaces.
\newblock {\em Transaction of the 7th Prague Conference on Information Theory,
  Statistical Decision Function and Random Processes\/}~{\em A}, 591--602.

\end{thebibliography}
\bibliographystyle{Chicago}

\end{document}